\documentclass{article}
\usepackage[utf8]{inputenc}
\usepackage{graphicx}
\usepackage{amsthm,amssymb,amsmath}
\usepackage{tikz}
\usetikzlibrary{fit}
\usetikzlibrary{arrows}
\usetikzlibrary{decorations.markings}
\usepackage{xspace}
\usepackage{multirow}
\usepackage{mathabx}
\usepackage{cancel}
\usepackage{varioref}
\usepackage{hyperref}
\usepackage[margin=1.5in]{geometry}
\usepackage{xcolor}
\usepackage{bm}
\usepackage[framemethod=TikZ]{mdframed}
\title{Simple, strict, proper, happy: \\ A study of reachability in temporal graphs\thanks{This work has been partially funded by the ANR project TEMPOGRAL (ANR-22-CE48-0001).}
}
\author{Arnaud Casteigts\thanks{University of Geneva, Switzerland.},~
        Timothée Corsini\thanks{LaBRI, CNRS, Univ. Bordeaux, Bordeaux INP, France.},
        Writika Sarkar\thanks{Chennai Mathematical Institute, India}}
\date{}
\newcommand{\G}{\ensuremath{\mathcal{G}}}
\newcommand{\closure}{\ensuremath{\mathcal{R}}}

\newtheorem{open}{Open question}
\newtheorem{lemma}{Lemma}
\newtheorem{claim}{Claim}
\counterwithin{claim}{lemma}
\newtheorem{theorem}{Theorem}
\newtheorem{corollary}{Corollary}
\newtheorem{definition}{Definition}
\newcommand{\ep}{\epsilon}

\newcommand{\minedge}{\textsc{Min-Edge Spanner}}
\newcommand{\minlabel}{\textsc{Min-Label Spanner}}

\definecolor{bluish}{RGB}{20, 22, 128}
\definecolor{prussian}{RGB}{0, 49, 83}
\definecolor{chocolate}{RGB}{101, 67, 33}
\definecolor{browny}{RGB}{150, 75, 0}
\definecolor{peach}{RGB}{255, 203, 164}
\definecolor{ocre}{RGB}{204, 119, 34}

\hypersetup{
  colorlinks=true,
  linkcolor=prussian,
  citecolor=prussian,
  filecolor=black,
  urlcolor=prussian,
}

\begin{document}
\maketitle

\begin{abstract}
  Dynamic networks are a complex subject. Not only do they inherit the
  complexity of static networks (as a particular case); they are also
  sensitive to definitional subtleties that are a frequent source
  of confusion and incomparability of results in the literature.

  In this paper, we take a step back and examine three such aspects in
  more details, exploring their impact in a systematic way;
  namely, whether the temporal paths are required to be \emph{strict} (i.e.,
  the times along a path must increasing, not just be non-decreasing),
  whether the time labeling is \emph{proper} (two adjacent edges
  cannot be present at the same time) and whether the time labeling is
  \emph{simple} (an edge can have only one presence time).
  In particular, we investigate how different combinations of these features impact
  the expressivity
  of the graph in terms of reachability.

  Our results imply a hierarchy of expressivity for the resulting
  settings, shedding light on the loss of generality that one is
  making when considering either combination. Some settings are more
  general than expected; in particular, proper temporal graphs turn
  out to be as expressive as general temporal graphs where non-strict
  paths are allowed. Also, we show that the simplest setting, that of
  \emph{happy} temporal graphs (i.e., both proper and simple) remains
  expressive enough to emulate the reachability of general temporal
  graphs in a certain (restricted but useful) sense. Furthermore, this
  setting is advocated as a target of choice for proving negative
  results. We illustrates this by strengthening two known results to
  happy graphs (namely, the inexistence of sparse spanners, and the
  hardness of computing temporal components). Overall, we hope that
  this article can be
  seen as a guide for choosing between different settings of temporal
  graphs, while being aware of the way these choices affect generality.
  \end{abstract}


\section{Introduction}
In the context of this paper, a temporal graph is a labeled graph
$\G=(V,E,\lambda)$ where $V$ is a finite set of vertices,
$E \subseteq V \times V$ a set of undirected edges, and
$\lambda : E \to 2^{\mathbb{N}}$ a function assigning at least one
time label to every edge, interpreted as presence times. These graphs
can model various phenomena, ranging from dynamic networks -- networks
whose structure changes over the time -- to dynamic interactions over
static (or dynamic) networks. These graphs have found applications in biology,
transportation, social networks, robotics, scheduling, distributed
computing, and self-stabilization, to name a few. Although more
complex formalisms have been defined and extensively studied (see
e.g.~\cite{CFQS12} or \cite{OR91}), several features of temporal
graphs remain not well understood even in the most restricted settings.

A fundamental aspect of temporal graphs is reachability,
commonly characterized in
terms of the existence of temporal paths; i.e., path which traverses
edges in chronological order. There has been a large number of studies
related to temporal reachability in the past two decades, seen from various
perspectives, e.g. $k$-connectivity and
separators~\cite{KKK02,FMN+20,CCMP20},
components~\cite{BF03,AF16,AGMS17,RML20,CLMS23,SSNP23}, feasibility of distributed
tasks~\cite{CFQS12,GCLL15,ADD+21,BT15}, schedule design~\cite{BCV21},
data structures~\cite{BFJ03,WDCG12,RML20,BACT22},
reachability minimization~\cite{EMMZ19}, reachability with additional
constraints~\cite{BFJ03,CHMZ21}, temporal spanners~\cite{AF16, AGMS17,
  CPS21, BDG+22}, path enumeration~\cite{EMM22}, random
graphs~\cite{BCF11,CRRZ21}, exploration~\cite{IKW14,DDFS20,ES22}, cops
and robbers~\cite{MW21,BLJ+23}, and
temporal flows~\cite{ACG+19,VDPS21}, to name a few (many more exist).
Over the course of these studies, it has become clear that temporal
connectivity differs significantly from classical reachability in
static graphs. To start with, it is not transitive, which implies that
two temporal paths (also called journeys) are not, in general,
composable, and consequently, connected components do not form
equivalence classes. This explains, in part, why many tractable
problems in static graphs become hard when transposed to temporal
graphs. Further complications arise, such as the conceptual impact of
having an edge appearing multiple times, and that of having adjacent
edges appearing at the same time. These aspects, while
innocent-looking, have a deep impact on the answers to many structural
and algorithmic questions.

In this paper, we take a step back, and examine the impact of such
aspects; in particular \emph{strictness} (should the times along a
path increase or only be non-decreasing?), \emph{properness} (can two
adjacent edges appear at the same time?) and \emph{simpleness} (do the
edges appear only once or several times?). We look at the impact of
these aspects from the point of view of temporal reachability, and
more precisely, how they restrict it. The central tool is the notion
of reachability graph, defined as the static directed graph where an
arc exists if and only if a temporal path exists in the original temporal
graph.\footnote{This concept was called the \emph{transitive
    closure of journeys} in~\cite{BF03,BACT22,CKNP19}; we now avoid
  this term because reachability is not transitive, which makes it
  somewhat misleading.} It turns out
that the above aspects have a strong impact on the kind of
reachability graph one can obtain from a temporal graph. Precisely, we
establish four separations between various combinations (called
\emph{settings}) of the above parameters. On the other hand, we also
present three reachability-preserving transformation between settings,
which show that certain settings are at least as expressive as others.

By combining the separations and transformations together with
arguments of containment, we obtain an almost complete hierarchy of
expressivity of these settings in terms of reachability. This
hierarchy clarifies the extent to which the choice of a particular
setting impacts generality, and as such, can be used as a guide for
future research in temporal graphs. Indeed, the above three aspects
(strictness, properness, simpleness) are a frequent source of
confusion and of incomparability of results in the literature.
Furthermore, many basic questions remain unresolved even in the most
restricted setting. For this reason, and somewhat paradoxically, we
advocate the study of the simplest model, that of \emph{happy}
temporal graphs (i.e., both proper and simple), where all the above
subtleties vanish. Another reason is that, despite being the least
expressive setting, happy graphs remain general enough to capture
\emph{certain features} of general temporal reachability. Finally,
negative results in this setting are \textit{de facto} stronger than
in all the other settings. In guise of illustration, we strengthen two
existing negative results to the happy setting. Namely, finding
temporal components of a given size remains difficult even in happy graphs
and the existence of
$o(n^2)$-sparse temporal spanners is also not guaranteed even in happy
graphs. Both results were initially obtained in more
general settings (respectively, in the non-proper, non-simple,
non-strict setting~\cite{BF03} for the former, and in the non-proper, simple,
non-strict setting~\cite{AF16} for the latter).

The paper is organized as follows. In Section~\ref{sec:definitions},
we give some definitions and argue that the above aspects deserve to
be studied for their own sake. In Section~\ref{sec:expressivity}, we
present the four separations and the three transformations, together
with the resulting hierarchy. In Section~\ref{sec:happy}, we
strengthen the hardness of temporal components and the counter-example
for sparse spanners to
the setting of happy graphs, and motivate their study further. Finally,
we conclude in Section~\ref{sec:conclusion} with some remarks.

\section{Temporal Graphs}
\label{sec:definitions}

Given a temporal graph $\G=(V,E,\lambda)$, the
static graph $G=(V,E)$ is called the \emph{footprint} of~$\G$.
Similarly, the static graph $G_t=(V,E_t)$ where
$E_t=\{e \in E \mid t \in \lambda(e)\}$ is the \emph{snapshot} of $\G$ at
time $t$. A pair $(e, t)$ such that $e \in E$ and $t \in \lambda(e)$
is a \emph{contact} (or temporal edge). The range of $\lambda$ is
called the \emph{lifetime} of $\G$, of length $\tau$. A \emph{temporal path} (or
journey) is a sequence of contacts $\langle(e_i, t_i)\rangle$ such
that $\langle e_i \rangle$ is a path in the footprint and $\langle t_i \rangle$
is non-decreasing.

The reachability relation based on temporal paths can be
captured by a \emph{reachability graph},
i.e. a static directed graph $\closure(\G)=(V, E_c)$, such that
$(u,v) \in E_c$ if and only if a temporal path exists from $u$ to $v$. A graph
$\G$ is \emph{temporally connected} if all the vertices can reach each
other at least once (i.e., $\closure(\G)$ is a complete directed graph).
The class of temporally connected graphs ($TC$)
is arguably one of the most basic classes of temporal graphs, along
with its infinite lifetime analog $TC^R$, where temporal
connectivity is achieved infinitely often (i.e., recurrently).

In what follows, we drop the adjective ``temporal'' whenever it is
clear from the context that the considered graph (or property) is temporal.

\subsection{Strictness / Properness / Simpleness}

The above definitions can be restricted in various ways. In
particular, one can identify three restrictions that are common in
the literature, although they are sometimes considered implicitly
and under various names:

\begin{itemize}
\item \textit{Strictness}: A temporal path $\langle(e_i, t_i)\rangle$
  is \emph{strict} if $\langle t_i\rangle$ is increasing.
\item \textit{Properness}: A temporal graph is \emph{proper} if
  $\lambda(e) \cap \lambda(e') = \emptyset$ whenever $e$ and $e'$ are
  incident to a same vertex (i.e., $\lambda$ is locally-injective).
\item \textit{Simpleness}: A temporal graph is \emph{simple} if
  $\lambda$ is single-valued; that is, every edge has a single
  presence time.
\end{itemize}

\emph{Strictness} is perhaps the easiest way of accounting for
traversal time for the edges. Without such restriction (i.e.,
in the default \emph{non-strict} setting), a journey can traverse arbitrarily
many edges at the same time step.
The notion of \emph{properness} is related to the one of strictness,
although not equivalent. Properness forces all the
journeys to be strict, because adjacent edges always have different
time labels. However, if the graph is non-proper, then considering strict
or non-strict journeys does have an impact, thus distinguishing both concepts
is important.
We will call \emph{happy} a graph that is both proper and simple, for
reasons that will become clear later.

Application-wise, proper temporal graphs arise naturally when
the graph represent mutually
exclusive interactions. Proper graphs also have the advantage that
$\lambda$ induces a proper coloring of the contacts (interpreting the
labels as colors). Finally, \emph{simpleness} naturally accounts for
scenarios where the entities interact only one time. It
is somewhat unlikely that a real-world system has this property;
however, this restriction has been extensively considered in
well studied subjects (e.g. in gossip theory).

Note that simpleness and properness are properties of the
\emph{graph}, whereas strictness is a property of the
\emph{temporal paths} in that graph. Therefore, one may either consider a strict or a
non-strict setting in a same temporal graph. The three notions (of
strictness, simpleness, and properness) interact in subtle ways, these
interactions being a frequent source of confusion and incomparability
among results. Before focusing on these interactions, let us make a
list of the possible combinations. The naive cartesian product of
these restrictions leads to eight combinations. However, not all of
them are meaningful, since properness removes the distinction between
strict and non-strict journeys. Overall, there are six
meaningful combinations, illustrated in Figure~\ref{fig:diagram}.

\begin{figure}[h]
  \begin{minipage}{6.5cm}
    {\small
\begin{itemize}
\item Non-proper, non-simple, strict (1)
\item Non-proper, non-simple, non-strict (2)
\item Non-proper, simple, strict (4)
\item Non-proper, simple, non-strict (5)
\item Proper, non-simple (3)
\item Proper, simple (= happy) (6)
\end{itemize}
}
\end{minipage}
\begin{minipage}{5cm}
  \centering
\colorlet{circle edge}{blue!50}
\colorlet{circle area}{blue!20}

\tikzset{
  filled/.style={fill = circle area, draw=circle edge, thick},
  outline/.style={draw=circle edge, thick}
}

\begin{tikzpicture}[scale=.5]

  \draw (0,0) circle[radius=3cm ]
  (0:4cm) circle[radius=3cm];
  \tikzstyle{every node}=[font=\footnotesize]
  \node [draw,fit=(current bounding box),inner sep=7mm] (frame){} ;

\node[text width=1cm, align=center] at (0:2cm) {Happy\\(6)};
\node[text width=1cm, align=center] at (-.5, 1) {Simple\\(4)};
\node[text width=1cm, align=center] at (-.5, -1) {(5)};
\node[text width=1cm, align=center] at (4.5, 1) {Proper\\(3)};
\node[text width=1cm] at (-2.7, 3) {Strict\\(1)};
\node[text width=1.4cm] at (-2.5,-3) {(2)\\Non-strict};

\draw (-4,0) -- (-3,0);
\draw[dashed] (-3,0) -- (1,0);
\draw (7,0) -- (8,0);
\end{tikzpicture}
  \end{minipage}
\caption{\label{fig:diagram}Settings resulting from combining the three properties.}
\end{figure}
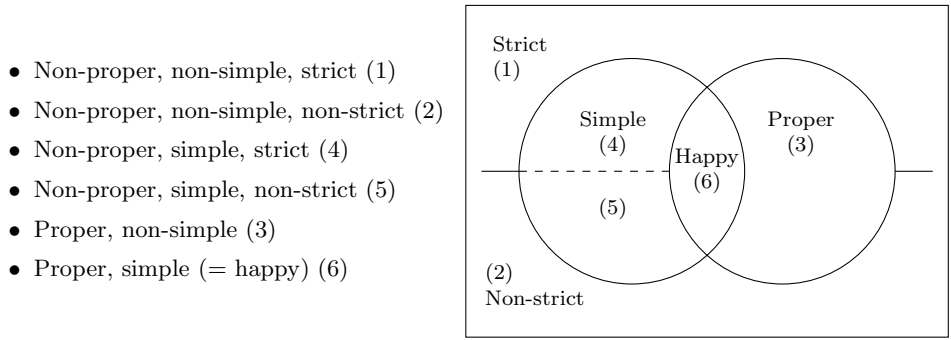

In the name of the settings, ``non-proper'' refers to the fact that properness is \emph{not required}, not to the fact that it is necessarily not satisfied. In other words, proper graphs are a particular case of non-proper graphs, and likewise, simple graphs are a particular case of non-simple graphs.
Thus, whenever non-proper or non-simple graphs are considered, we will omit this information from the name. For instance, setting (1) will be referred to as the (general) strict setting. Finally, observe that strict paths are a special case of non-strict paths, but we do not get an inclusion of the corresponding settings, whose features are actually incomparable (we shall return on that subtle point later).

\subsection{Does it really matter? (Example of spanners)}

While innocent-looking, the choice for a particular setting may have
tremendous impacts on the answers to basic questions. For illustration, consider the
spanner problem. Given a graph $\G=(V,E,\lambda)$ such that
$\G \in TC$, a \emph{temporal spanner} of $\G$ is a graph
$\G'=(V,E',\lambda')$ such that $\G' \in TC$, $E' \subseteq E$, and
for all $e$ in $E'$, $\lambda'(e) \subseteq \lambda(e)$. In other
words, $\G'$ is a temporally connected spanning subgraph of $\G$.
A natural goal is to minimize the size of the spanner, either in terms of
number of labels or number of underlying edges. More formally,

\newenvironment{problem}[1]
{ \begin{mdframed}[roundcorner=2pt, linecolor=prussian!100, linewidth=1pt, backgroundcolor=prussian!10] \interlinepenalty=10000\color{prussian} #1 \\}
  {  \end{mdframed}}

\begin{problem}{\minlabel}
  Input: A temporal graph $\G$, an integer $k$\\
  Output: Does $\G$ admit a temporal spanner with at most $k$ contacts?
\end{problem}

\begin{problem}{\minedge}
  Input: A temporal graph $\G$, an integer $k$\\
  Output: Does $\G$ admit a temporal spanner of at most $k$ edges (keeping all their labels)?
\end{problem}
\medskip

The search and optimization versions of these problems can be defined
analogously. Unlike spanners in static graphs, the
definition does not care about stretch factors, due to the fact that
the very existence of small spanners is not guaranteed.
In the following, we illustrate the impact of the
notions of strictness, simpleness, and properness (and their
interactions) on these questions. The impact of strictness is pretty
straightforward. Consider the graph $\G_1$ on
Figure~\ref{fig:spanners}. If non-strict journeys are allowed, then
this graph admits $\G_2$ as a spanner (among others), this spanner being optimal for
both versions of the problem (3 labels, 3 edges). Otherwise, the
minimum spanners are bigger (and different) for both versions: $\G_3$
minimizes the number of labels (4 labels, 4 edges), while $\G_4$
minimizes the number of edges (3 edges, 5 labels). If strictness is
combined with non-properness, then there exist a pathological scenario
(already identified in~\cite{KKK02} and~\cite{AGMS17}) where the input is a complete
temporal graph (see $\G_5$, for example) but none of the edges
can be removed without breaking connectivity! Note that
$\G_5$ is a simple temporal graph. Simpleness has further consequences.
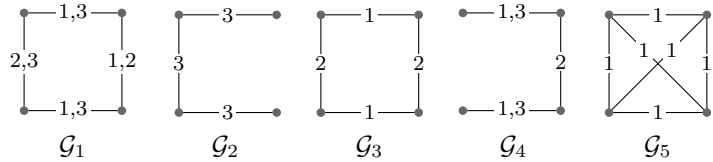
\begin{figure}[ht]
  \centering
  \tikzset{defnode/.style={draw, circle, fill, darkgray!80, inner sep=1pt}}
  \begin{tabular}{ccccc}
\begin{tikzpicture}[scale=1.3]
  \tikzstyle{every node}=[defnode]
  \path (0,1) node (a){};
  \path (1,1) node (b){};
  \path (1,0) node (c){};
  \path (0,0) node (d){};

  \tikzstyle{every node}=[font=\footnotesize,fill=white,inner sep=1pt]
  \draw (a) -- node{1,3} (b);
  \draw (a) -- node{2,3} (d);
  \draw (b) -- node{1,2} (c);
  \draw (c) -- node{1,3} (d);
\end{tikzpicture}&
\begin{tikzpicture}[scale=1.3]
  \tikzstyle{every node}=[defnode]
  \path (0,1) node (a){};
  \path (1,1) node (b){};
  \path (1,0) node (c){};
  \path (0,0) node (d){};

  \tikzstyle{every node}=[font=\footnotesize,fill=white,inner sep=1pt]
  \draw (a) -- node{3} (b);
  \draw (a) -- node{3} (d);
  \draw (c) -- node{3} (d);
\end{tikzpicture}&
\begin{tikzpicture}[scale=1.3]
  \tikzstyle{every node}=[defnode]
  \path (0,1) node (a){};
  \path (1,1) node (b){};
  \path (1,0) node (c){};
  \path (0,0) node (d){};

  \tikzstyle{every node}=[font=\footnotesize,fill=white,inner sep=1pt]
  \draw (a) -- node{1} (b);
  \draw (a) -- node{2} (d);
  \draw (b) -- node{2} (c);
  \draw (c) -- node{1} (d);
\end{tikzpicture}&
\begin{tikzpicture}[scale=1.3]
  \tikzstyle{every node}=[defnode]
  \path (0,1) node (a){};
  \path (1,1) node (b){};
  \path (1,0) node (c){};
  \path (0,0) node (d){};

  \tikzstyle{every node}=[font=\footnotesize,fill=white,inner sep=1pt]
  \draw (a) -- node{1,3} (b);
  \draw (b) -- node{2} (c);
  \draw (c) -- node{1,3} (d);
\end{tikzpicture}&
\begin{tikzpicture}[scale=1.3]
  \tikzstyle{every node}=[defnode]
  \path (0,1) node (a){};
  \path (1,1) node (b){};
  \path (1,0) node (c){};
  \path (0,0) node (d){};

  \tikzstyle{every node}=[font=\footnotesize,fill=white,inner sep=1pt]
  \draw (a) -- node{1} (b);
  \draw (a) -- node{1} (d);
  \draw (b) -- node{1} (c);
  \draw (c) -- node{1} (d);
  \draw (b) -- node[pos=.35]{1} (d);
  \draw (a) -- node[pos=.35]{1} (c);
\end{tikzpicture}\\
    $\G_1$&$\G_2$&$\G_3$&$\G_4$&$\G_5$\\
  \end{tabular}
\caption{\label{fig:spanners}Some temporal graphs on four vertices.}
\end{figure}
For example, if the input graph is simple and proper, then it cannot
admit a spanning tree (i.e. a spanner of $n-1$ edges) and requires at
least $2n-4$ edges (or labels, equivalently, since the graph is
simple)~\cite{B79}. If the input graph is simple and non-proper, then
it does not admit a spanning tree if strictness is required, but it
does admit one otherwise if and only if at least one of the snapshots
is connected (in a classical sense). Finally, none of these affirmations hold in
general for non-simple graphs.

If the above discussion seems confusing to the reader, it is not
because we obfuscated it. The situation is intrinsically subtle. In
particular, one should bear in mind the above subtleties whenever
results from different settings are compared with each other. To
illustrate such pitfalls, let us relate a recent mistake (fortunately,
without consequences) that involved one of the authors.
In~\cite{AF16}, Axiotis and Fotakis constructed a (non-trivial)
infinite family of temporal graphs which do not admit $o(n^2)$-sparse
spanners. Their construction is given in the setting of simple
temporal graphs, with non-proper labeling and non-strict journeys
allowed. The same paper actually uses many constructions formulated in
this setting, and a general claim is that these constructions can be
adapted to proper graphs (and so strict journeys). Somewhat hastily,
the introductions of~\cite{CRRZ21} and~\cite{CPS21} claim that the
counterexample from~\cite{AF16} holds in happy graphs. The pitfall is that, for
some of the constructions in~\cite{AF16}, giving up on non-properness
(and non-strictness) is only achievable at the cost of using
\emph{multiple} labels per edge -- a conclusion that we reached
in the meantime. To be fair, the authors of~\cite{AF16} never
claimed that these adaptations could preserve simpleness, so their
claim was actually correct.

Apart from illustrating the inherent subtleties of these notions, the
previous observations imply that counter-examples to sparse spanners
in happy graphs was in fact still open. In
Section~\ref{sec:counterexample}, we show that the spanner
construction from~\cite{AF16} can indeed be adaptated to this very
restricted setting.

\subsection{Happy Temporal Graphs}

A temporal graph $\G=(V,E,\lambda)$ is \emph{happy} if it is both
proper and simple. These graphs have sometimes been referred to as
\emph{simple temporal graphs} (including by the authors), which the
present paper now argues is insufficiently precise. Happy graphs are
``happy'' for a number of reasons. First, the distinction between
strict journeys and non-strict journeys can be safely ignored (due to
properness), and the distinction between contacts and edges can also
be ignored (due to simpleness). Clearly, these restrictions come with
a loss of expressivity, but this does not prevent happy graphs from
being relevant more generally in the sense that negative results in
these graphs carry on to all the other settings. For
example, if a problem is computationally hard
on happy instances, then it is so in all the other settings. Thus, it seems
like a good
practice to try to prove negative results for happy graphs first, whenever
possible. If this is not possible, then proving it in proper graphs
still has the advantage of making it applicable to both strict
and non-strict temporal paths alike.
Positive results, on the other hand, are not generally
transferable; in particular, a hard problem in general temporal graphs
could become tractable in happy graphs. This being said, if a certain
graph contains a happy subgraph, then whatever pattern can be found in
the latter also exists in the former, which enables \emph{some form}
of transferability for positive results as well from happy graphs
to more general temporal graphs.

In fact, happy graphs coincide with a vast body of literature. Many
studies in \emph{gossip theory} and \emph{population protocols}
consider the same restrictions, and the so-called \emph{edge-ordered
  graphs}~\cite{CK71} can also be seen as a particular case of happy
graphs where $\lambda$ is \emph{globally} injective (although the
distinction does not matter for reachability). In addition, a number
of other existing results in temporal graphs consider such restrictions.

Finally, a nice property of happy graph is that, up to time-distortion that preserve the local ordering of the
edges, the number of happy graphs on a certain
number of vertices is \emph{finite} -- a crucial
property for exhaustive search and verification (note that this is also the case
of simple graphs, more generally).

We think that the above arguments, together with the fact that many basic questions
remain unsolved even in this restricted model, makes happy graphs a
compelling class of temporal graphs to be studied in the current state
of knowledge.

\section{Expressivity of the settings in terms of reachability}
\label{sec:expressivity}

As already said, a fundamental aspect of temporal graphs is
reachability through temporal paths. There are
several ways of characterizing the extent to which two temporal graphs
$\G_1$ and $\G_2$ have similar reachability. The first three, below,
are increasingly more restrictive.

\begin{definition}[Reachability equivalence]
  Let $\G_1$ and $\G_2$ be two temporal graphs built on the same set of
  vertices. $\G_1$ and $\G_2$ are
  \emph{reachability-equivalent} if $\closure(\G_1) \simeq \closure(\G_2)$
  (i.e. both reachability graphs are isomorphic). By abuse of language, we
  say that $\G_1$ and $\G_2$ have the ``same'' reachability graph.
\end{definition}

\begin{definition}[Support equivalence]
  Let $\G_1$ and $\G_2$ be two temporal graphs built on the same set of
  vertices. These graphs are \emph{support-equivalent} if for
  every journey in either graph, there exists a journey in the other
  graph whose underlying path goes through the same sequence of
  vertices.
\end{definition}

\begin{definition}[Bijective equivalence]
  Let $\G_1$ and $\G_2$ be two temporal graphs built on a same set of
  vertices. These graphs are \emph{bijectively equivalent} if
  there is a bijection $\sigma$ between the set of journeys of
  $\G_1$ and that of $\G_2$, and $\sigma$ is
  support-preserving.
\end{definition}

The following form of equivalence is weaker.

\begin{definition}[Induced reachability equivalence]
  Let $\G_1$ and $\G_2$ be two temporal graphs built on vertices $V_1$
  and $V_2$, respectively, with $V_1 \subseteq V_2$. $\G_2$ is
  \emph{induced-reachability equivalent} to $\G_1$ if
  $\closure(\G_2)[V_1] \simeq \closure(\G_1)$. In other words, the
  restriction of $\closure(\G_2)$ to the vertices of $V_1$ is
  isomorphic to $\closure(\G_1)$.
\end{definition}

Observe that
bijective equivalence implies support equivalence, which implies
reachability equivalence, which implies induced reachability equivalence.
Furthermore, support equivalence forces both footprints
to be the same (the converse is not true).
In this section, we show that some of the settings differ in terms of reachability, whereas others
coincide. We first prove a number of
\emph{separations}, by showing that there exist temporal graphs in
some setting, whose reachability graph cannot be realized in some
other settings (Section~\ref{sec:separations}). Then, we present three
\emph{transformations} which establish various levels of equivalences
 (Section~\ref{sec:transformations}).
Finally, we infer more relations
by combining separations and transformations in
Section~\ref{sec:summary}, together with further discussions. A
complete diagram illustrating all the relations is given in the end of
the section (Figure~\vref{fig:summary}).

\subsection{Separations}
\label{sec:separations}

In view of the above discussion, a separation in terms of reachability
graphs is pretty general, as it implies a separation
for the two stronger forms of equivalences (support-preserving and
bijective ones). Before starting, let us state a simple lemma
used in several of the subsequent proofs.

\begin{lemma}
  \label{lem:distance-3}
  In the non-strict setting, if two vertices are at distance
  two in the footprint, then at least one of them can reach the other
  (i.e. the reachability graph must have at least one arc between these
  vertices).
\end{lemma}

\subsubsection{``Simple \& strict'' $vs.$ ``strict''}

\begin{lemma}
  \label{lemma:ft-strict-simple}
  There is a graph in the ``strict'' setting whose reachability graph cannot be obtained from a graph in the ``simple \& strict'' setting.
\end{lemma}

\begin{proof}
  Consider the following non-simple graph $\G$ (left) in a strict setting and the corresponding reachability graph (right). We will prove that a hypothetical simple temporal graph $\mathcal{H}$ with same reachability graph as $\G$ cannot be built in the strict setting.
  \begin{figure}[ht]
    \centering
    \begin{tabular}{cc}
    $\G=$
    \begin{tikzpicture}
      \tikzset{vertex/.style={draw, circle, inner sep=0.55mm}}

      \node[vertex] (a) [label = $a$] at (0,0) {};
      \node[vertex] (b) [label = $b$] at (1,0) {};
      \node[vertex] (c) [label = $c$] at (2,0) {};
      \node[vertex] (d) [label = $d$] at (3,0) {};

      \begin{scope}[font=\footnotesize]
        \draw (a) edge node[above] {1} (b);
        \draw (b) edge node[above] {1,2} (c);
        \draw (c) edge node[above] {2} (d);
      \end{scope}
    \end{tikzpicture}
      &
        $\closure(\G)=$
    \begin{tikzpicture}
      \tikzset{vertex/.style={draw, circle, inner sep=0.55mm}}

      \node[vertex] (a) [label = $a$] at (0,0) {};
      \node[vertex, above] (b) [label = $b$] at (1,0) {};
      \node[vertex, below] (c) [label = $c$] at (2,0) {};
      \node[vertex] (d) [label = $d$] at (3,0) {};

      \begin{scope}[font=\footnotesize]
        \draw (a) edge (b);
        \draw (b) edge (c);
        \draw (c) edge (d);
        \draw[->,semithick,bend right] (a) edge (c);
        \draw[->,semithick,bend left] (b) edge (d);
      \end{scope}
    \end{tikzpicture}
      \\
      \end{tabular}
  \end{figure}
  First, observe that the arc $(a,c)$ in $\closure(\G)$ exists only in one direction. Thus, $a$ and $c$ cannot be neighbors in $\mathcal{H}$. Since $\mathcal{H}$ is simple and the journeys are strict (and $a$ has no other neighbors in $\closure(\G)$), the arc $(a,c)$ can only result from the label of $ab$ being strictly less than $bc$. The same argument holds between $bc$ and $cd$ with respect to the arc $(b,d)$ in $\closure(\G)$. As a result, the labels of $ab$, $bc$, and $cd$ must be strictly increasing, which is impossible since $(a,d)$ does not exist in $\closure(\G)$.
\end{proof}

As simple graphs are a particular case of non-simple graphs, the following follows.

\begin{corollary}
  \label{cor:simple-strict<simple-non-strict}
  The ``simple \& strict'' setting is strictly less expressive than the ``strict'' setting in terms of reachability graphs.
\end{corollary}

\subsubsection{``Non-strict'' $vs.$ ``simple \& strict''}

\begin{lemma}
  \label{lemma:ft-strict-non-strict}
  There is a graph in the ``simple \& strict'' setting whose reachability graph cannot be obtained from a graph in the ``non-strict'' setting.
\end{lemma}

\begin{proof}
  Consider the following simple temporal graph $\G$ (left) in a strict
  setting and the corresponding reachability graph (right). Note that
  $a$ and $c$ are not neighbors in $\closure(\G)$, due to strictness. For
  the sake of contradiction, let $\mathcal{H}$ be a temporal graph
  whose non-strict reachability graph is isomorphic to that of $\G$.
  \begin{figure}[ht]
    \centering
    \begin{tabular}{cc}
      $\G=$
      \begin{tikzpicture}[every loop/.style={}, thick, node distance=1.2cm,auto]
        \tikzset{vertex/.style={draw, circle, inner sep=0.55mm}}

        \node[vertex] (a) [label = $a$] at (0,0) {};
        \node[vertex] (b) [label = $b$] at (1,0) {};
        \node[vertex] (c) [label = $c$] at (2,0) {};

        \begin{scope}
        \draw (a) edge node {1} (b);
        \draw (b) edge node {1} (c);
        \end{scope}
    \end{tikzpicture}
    & $\closure(\G)=$
    \begin{tikzpicture}[every loop/.style={}, thick, node distance=1.2cm,auto]
        \tikzset{vertex/.style={draw, circle, inner sep=0.55mm}}

        \node[vertex] (a) [label = $a$] at (0,0) {};
        \node[vertex] (b) [label = $b$] at (1,0) {};
        \node[vertex] (c) [label = $c$] at (2,0) {};

        \begin{scope}
        \draw (a) edge node {} (b);
        \draw (b) edge node {} (c);
        \end{scope}
    \end{tikzpicture}
    \\
    \end{tabular}
    \label{fig:strict-counter}
  \end{figure}
  First, observe that the footprint of $\mathcal{H}$ must be
  isomorphic to the footprint of $\G$, as otherwise it is either complete or not
  connected. Call $b$ the vertex of degree two in
  $\mathcal{H}$. If
  $\lambda_{\mathcal{H}}(ab) \ne \lambda_{\mathcal{H}}(bc)$, then
  either $a$ can reach $c$ or $c$ can reach $a$, and if
  $\lambda_{\mathcal{H}}(ab) = \lambda_{\mathcal{H}}(bc)$, then both can reach each other through a non-strict
  journey. In both cases, $\closure(\mathcal{H})$ contains more arcs
  than $\closure(\mathcal{G})$.
\end{proof}

As stated in the end of the section, the reverse direction is left open.

\subsubsection{``Simple \& non-strict'' $vs.$ ``proper''}

\begin{lemma}
  \label{lemma:ft-non-strict-simple}
  There is a graph in the ``proper'' setting whose reachability graph cannot be obtained from a graph in the ``simple \& non-strict'' setting.
\end{lemma}

\begin{proof}
  Consider the following proper temporal graph $\G$ (left). Its reachability graph (right) is a graph on four vertices, with an edge between any pair of vertices except $a$ and $d$ (i.e., a diamond). For the sake of contradiction, let $\mathcal{H}$ be a simple temporal graph in the non-strict setting, whose reachability graph is isomorphic to that of~$\G$.
  \begin{figure}[ht]
    \centering
    \begin{tabular}{cc}
    $\G=$
    \begin{tikzpicture}
      \tikzset{vertex/.style={draw, circle, inner sep=0.55mm}}

      \node[vertex] (a) [label = $a$] at (0,0) {};
      \node[vertex] (b) [label = $b$] at (1,0) {};
      \node[vertex] (c) [label = $c$] at (2,0) {};
      \node[vertex] (d) [label = $d$] at (3,0) {};

      \begin{scope}[font=\footnotesize]
        \draw (a) edge node[above] {2} (b);
        \draw (b) edge node[above=-1.5pt] {1,3} (c);
        \draw (c) edge node[above] {2} (d);
      \end{scope}
    \end{tikzpicture}
      &
        $\closure(\G)=$
    \begin{tikzpicture}
      \tikzset{vertex/.style={draw, circle, inner sep=0.55mm}}

      \node[vertex] (a) [label = $a$] at (0,0) {};
      \node[vertex,above] (b) [label = $b$] at (1,0) {};
      \node[vertex,below] (c) [label = $c$] at (2,0) {};
      \node[vertex] (d) [label = $d$] at (3,0) {};

      \begin{scope}[font=\footnotesize]
        \draw (a) edge (b);
        \draw (b) edge (c);
        \draw (c) edge (d);
        \draw[bend right] (a) edge (c);
        \draw[bend left] (b) edge (d);
      \end{scope}
    \end{tikzpicture}
      \\
      \end{tabular}
  \end{figure}
  First, observe that no arcs exist between $a$ and $d$ in the
  reachability graph, thus $a$ and $d$ must be at least at distance
  $3$ in the footprint (Lemma~\ref{lem:distance-3}), which is only
  possible if the footprint is a graph isomorphic to $P_4$ (i.e. a
  path graph on four vertices) with endpoints $a$ and $d$.
  Now, let $t_1,t_2,$ and $t_3$ be the labels of $ab,bc$, and $cd$ respectively.
  Since
  $\{a,b,c\}$ is a clique in the reachability graph (whatever the way identifiers
  $b$ and $c$ are assigned among the two remaining vertices), they
  must be temporally connected in $\mathcal{H}$, which forces that
  $t_1 = t_2$ (otherwise both edges could be travelled in only one
  direction). Similarly, the fact that $\{b,c,d\}$ is a clique in the
  reachability graph forces $t_2 = t_3$. As a result, there must be a non-strict
  journey between $a$ and $d$, which contradicts the
  absence of arc between $a$ and $d$ in the reachability graph.
\end{proof}

\noindent
The next corollary follows by inclusion of proper graphs in the non-strict
setting.

\begin{corollary}
  \label{cor:simple<non-strict}
  The ``simple \& non-strict'' setting is strictly less expressive than the ``non-strict'' setting in terms of reachability graphs.
\end{corollary}

\subsubsection{``simple \& proper (i.e. happy)'' $vs.$ ``simple \& non-strict''}

\begin{lemma}
  \label{lemma:ft-non-strict-simple-happy}
  There is a graph in the ``simple \& non-strict'' setting whose reachability graph cannot be obtained in the ``happy'' setting.
\end{lemma}

\begin{proof}
  Consider the following simple temporal graph $\G$ (left) in a non-strict
  setting and the corresponding reachability graph (right). For the
  sake of contradiction, let $\mathcal{H}$ be a happy temporal graph
  whose reachability graph is isomorphic to that of $\mathcal{G}$.

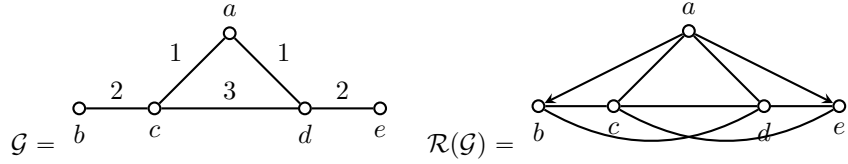
\begin{figure}[ht]
  \centering
  \begin{tabular}{cc}
    $\G=$
    \begin{tikzpicture}[every loop/.style={}, thick, node distance=1.2cm,auto]
        \tikzset{vertex/.style={draw, circle, inner sep=0.55mm}}

        \node[vertex] (a) [label = $a$] at (0,0) {};
        \node[vertex] (b) [label = below:$b$] at (-2,-1) {};
        \node[vertex] (c) [label = below:$c$] at (-1,-1) {};
        \node[vertex] (d) [label = below:$d$] at (1,-1) {};
        \node[vertex] (e) [label = below:$e$] at (2,-1) {};

        \begin{scope}[>=stealth]
          \draw (c) edge node {$1$} (a);
          \draw (a) edge node {$1$} (d);
          \draw (b) edge node {$2$} (c);
          \draw (c) edge node {$3$} (d);
          \draw (d) edge node {$2$} (e);
        \end{scope}
    \end{tikzpicture}

    & $\closure(\G)=$

    \begin{tikzpicture}[every loop/.style={}, thick, node distance=1.2cm,auto]
        \tikzset{vertex/.style={draw, circle, inner sep=0.55mm}}

        \node[vertex] (a) [label = $a$] at (0,0) {};
        \node[vertex] (b) [label = below:$b$] at (-2,-1) {};
        \node[vertex] (c) [label = below:$c$] at (-1,-1) {};
        \node[vertex] (d) [label = below:$d$] at (1,-1) {};
        \node[vertex] (e) [label = below:$e$] at (2,-1) {};

        \begin{scope}[>=stealth]
          \draw (c) edge node {} (a);
          \draw (a) edge node {} (d);
          \draw (b) edge node {} (c);
          \draw (c) edge node {} (d);
          \draw (d) edge node {} (e);

          \draw (a) edge[->] node {} (b);
          \draw (a) edge[->] node {} (e);
          \draw (b) edge[bend right] node {} (d);
          \draw (c) edge[bend right] node {} (e);
        \end{scope}
    \end{tikzpicture}
    \\
  \end{tabular}
  \label{fig:nstrict-simple-counter}
\end{figure}
Since $a$ is not isolated in the reachability graph, it has at least one
neighbor in~$\mathcal{H}$. Vertices $b$ and $e$ cannot be such
neighbors, the arc being oneway in the reachability graph, so its
neighbors are either $c$, $d$, or both $c$ and $d$. \textit{Wlog},
assume that $c$ is a neighbor (the arguments hold symmetrically for
$d$), we first prove an intermediate statement
\begin{claim}
  \label{claim:no-bd}
  The edge $bd$ does not exists in the footprint of $\mathcal{H}$.
\end{claim}
\begin{proof}[Proof of Claim~\ref{claim:no-bd} (by contradition)]
  If $bd \in \mathcal{H}$, then $de \notin \mathcal{H}$, as
  otherwise $b$ and $e$ would be at distance $2$ and share at least
  one arc in the reachability graph (Lemma~\ref{lem:distance-3}).
  However, $e$ must have at least one neighbor, thus
  $ce \in \mathcal{H}$, and by Lemma~\ref{lem:distance-3} again
  $bc \notin \mathcal{H}$. At this point, the footprint of $\mathcal{H}$
  must look like the following graph, in which the status of $ad$
  and $cd$ is not settled yet.
  
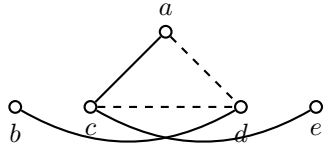
\begin{center}
    \begin{tikzpicture}[every loop/.style={}, thick, node distance=1.2cm,auto]
      \tikzset{vertex/.style={draw, circle, inner sep=0.55mm}}

      \node[vertex] (a) [label = $a$] at (0,0) {};
      \node[vertex] (b) [label = below:$b$] at (-2,-1) {};
      \node[vertex] (c) [label = below:$c$] at (-1,-1) {};
      \node[vertex] (d) [label = below:$d$] at (1,-1) {};
      \node[vertex] (e) [label = below:$e$] at (2,-1) {};

      \begin{scope}[>=stealth]
        \draw (c) edge node {} (a);
        \draw (a) edge[dashed] node {} (d);
        \draw (c) edge[dashed]  node {} (d);

        \draw (b) edge[bend right] node {} (d);
        \draw (c) edge[bend right] node {} (e);
      \end{scope}
    \end{tikzpicture}
  \end{center}
  In fact, $ad$ must exist, as otherwise there is no way of connecting
  $d$ to $a$ \emph{and} $a$ to~$d$. Also note that the absence of
  $(e,a)$ in the reachability graph forces $\lambda(ac) < \lambda(ce)$ (remember
  that $\mathcal{H}$ is both proper and simple), which implies that no
  journey exists from $e$ to $d$ unless $cd$ is also added to
  $\mathcal{H}$ with a label $\lambda(cd) > \lambda(ce)$. In the
  opposite direction, $d$ needs that $\lambda(ad) < \lambda(ac)$ to be
  able reach $e$. Now, $c$ needs that $\lambda(cd) < \lambda(bd)$
  to reach $b$. In summary, we must have
  $\lambda(ad) < \lambda(ac) < \lambda(ce) < \lambda(cd) < \lambda(bd)$,
  which implies that $b$ cannot reach~$c$.
\end{proof}

By this claim, $bd \notin \mathcal{H}$, thus
$bc \in \mathcal{H}$ and consequently $cd \notin \mathcal{H}$ (by
Lemma~\ref{lem:distance-3}). From the absence of $(b,a)$ in the
reachability graph, we infer that $\lambda(bc) > \lambda(ac)$. In order for $b$
to reach $d$, we need that $cd$ exists with label
$\lambda(cd) > \lambda(bc)$. To make $d$ to $b$ mutually reachable,
there must be an edge $ad$ with time $\lambda(ad) < \lambda(ac)$.
Now, the only way for $c$ to reach $e$ is through the edge $de$,
and since there is no arc $(e,a)$, its label must satisfy
$\lambda(de) > \lambda(ad)$. Finally, $c$ can reach $e$ (but not
through $a$), so $\lambda(de) > \lambda(cd)$ and $c$
cannot reach $e$, a contradiction.
\end{proof}

By inclusion of happy graphs in the ``simple \& non-strict'' setting, we have

\begin{corollary}
  \label{cor:happy<simple-non-strict}
  The ``simple \& proper (i.e. happy)'' setting is strictly less expressive than the ``simple \& non-strict'' setting in terms of reachability graphs.
\end{corollary}

\subsection{Transformations}
\label{sec:transformations}

In this section, we present three transformations. First, we present a
transformation from the general non-strict setting to the
setting of proper graphs, called the \emph{dilation technique}.
Since proper graphs are contained in both the non-strict and
strict settings, this transformation implies that the strict setting
is at least as expressive as the non-strict setting. This transformation is
\emph{support-preserving}, but it suffers from a significant
blow-up in the size of the lifetime. Another transformation called the
\emph{saturation technique} is presented from the (general) non-strict
setting to the (general) strict setting, which is only
\emph{reachability-preserving} but preserves the size of the lifetime.
Finally, we present an induced-reachability-preserving transformation, called
the \emph{semaphore technique}, from
the general strict setting to happy graphs. If the original
temporal graph is non-strict, one can compose it with one of the
first two transformations, implying that \emph{all} temporal graphs
can be turned into a happy graph whose reachability graph
contains that of the original temporal graph as an induced
subgraph. This shows that happy graphs are universal in a weak (in fact, induced) sense.

\subsubsection{Dilation technique: ``non-strict'' $\to$ ``proper''}

Given a temporal graph $\mathcal{G}$ in the non-strict setting, we present a transformation that creates a proper temporal graph $\mathcal{H}$ that is support-equivalent to $\G$ (and thus also reachability-equivalent).
We refer to this transformation as the \emph{dilation} technique.

The transformation operates at the level of the snapshots, taken independently, one after the other. It consists of isolating, in turn, every snapshot $G_t$ where some non-strict journeys are possible, and ``dilating'' it over more time steps in such a way these journeys can be made strict (note that this needs be applied only if $G_t$ contains at least one path of length larger than $1$). The subsequent snapshots are shifted in time accordingly.
The dilation of a snapshot $G_t$ goes as follows. Without loss of generality, assume that $t=1$ (otherwise, shift the labels used below by the sum of lifetimes resulting from the dilation of the earlier snapshots). First, we transform $G_t$ into a non-proper temporal graph $\G_t$ whose footprint is $G_t$ itself, and the edges of which are assigned labels $1,2,..., k$, where $k$ is the longest path in $G_t$ (with $k \le |V| - 1$). As argued in the proof below, there is a strict journey in $\G_t$ if and only if there is a path (and thus a non-strict journey) in $G_t$.
Now, $\G_t$ can be turned into a proper graph as follows.
By Vizing's theorem, the edges of a graph of maximum degree $\Delta$ can be properly colored using at most $\Delta + 1$ colors. Let $c:E_t \to [0,\Delta]$ be such a coloring, and let $\epsilon$ be a fixed value less than $1/(\Delta+1)$. Each label of each edge $e$ of $E_t$ is ``tilted'' in $\G_t$ by a quantity equal to $c(e)\epsilon$.
The transformation is illustrated in Figure~\ref{fig:shift-method}.

\begin{figure}[ht]
  \centering
  \begin{tabular}{cccc}
    \begin{tikzpicture}[xscale=.8,every loop/.style={}, thick, node distance=1.2cm,auto]
        \tikzset{vertex/.style={draw, circle, inner sep=0.55mm}}
        \tikzstyle{path} = [draw,line width=2pt, dashdotted,-,red]

        \node[vertex] (a) at (2,2) {}; 
        \node[vertex] (b) at (0,2) {};
        \node[vertex] (c) at (0,0) {};
        \node[vertex] (e) at (2,0) {};
        \node[vertex] (f) at (1,-0.5) {};

        \begin{scope}
        \draw (b) edge[above, sloped]  node {\footnotesize $1,3$} (a);
        \draw (c) edge[above, sloped] node {\footnotesize $1,2$} (b);
        \draw (c) edge node {\footnotesize $1$} (e);
        \draw (f) edge[below, sloped] node {\footnotesize $3,4$} (c);
        \draw (e) edge node {\footnotesize $2$} (f);

        \end{scope}
    \end{tikzpicture}
    &
    \begin{tikzpicture}[xscale=.8, every loop/.style={}, thick, node distance=1.2cm,auto]
        \tikzset{vertex/.style={draw, circle, inner sep=0.55mm}}
        \tikzstyle{path} = [draw,line width=2pt, dashdotted,-,red]

        \node[vertex] (a) at (2,2) {}; 
        \node[vertex] (b) at (0,2) {};
        \node[vertex] (c) at (0,0) {};
        \node[vertex] (e) at (2,0) {};
        \node[vertex] (f) at (1,-0.5) {};
        \path (f)+(0,-.3) coordinate (bidon);

        \begin{scope}
          \draw (b) edge[above, sloped] (a);
          \draw (c) edge[above, sloped] (b);
          \draw (c) edge[above, sloped] (e);


        \end{scope}
    \end{tikzpicture}
    &
    \begin{tikzpicture}[xscale=.8, every loop/.style={}, thick, node distance=1.2cm,auto]
        \tikzset{vertex/.style={draw, circle, inner sep=0.55mm}}
        \tikzstyle{path} = [draw,line width=2pt, dashdotted,-,red]

        \node[vertex] (a) at (2,2) {}; 
        \node[vertex] (b) at (0,2) {};
        \node[vertex] (c) at (0,0) {};
        \node[vertex] (e) at (2,0) {};
        \node[vertex] (f) at (1,-0.5) {};
        \path (f)+(0,-.3) coordinate (bidon);

        \begin{scope}
          \draw (b) edge[above, sloped] node {\footnotesize $1,2,3$} (a);
          \draw (c) edge[above, sloped] node {\footnotesize $1$$+$$\epsilon,\ 2$$+$$\epsilon,\ 3$$+$$\epsilon$} (b);
          \draw (c) edge[above, sloped] node {\footnotesize $1,2,3$} (e);


        \end{scope}
    \end{tikzpicture}
    &
    \begin{tikzpicture}[xscale=.8, every loop/.style={}, thick, node distance=1.2cm,auto]
        \tikzset{vertex/.style={draw, circle, inner sep=0.55mm}}
        \tikzstyle{path} = [draw,line width=2pt, dashdotted,-,red]

        \node[vertex] (a) at (2,2) {}; 
        \node[vertex] (b) at (0,2) {};
        \node[vertex] (c) at (0,0) {};
        \node[vertex] (e) at (2,0) {};
        \node[vertex] (f) at (1,-0.5) {};

        \begin{scope}
        \draw (b) edge[above, sloped] node {\footnotesize  $1,3,5,8$} (a);
        \draw (c) edge[below, sloped] node {\footnotesize  $2,4,6,7$} (b);
        \draw (c) edge[above, sloped] node {\footnotesize  $1,3,5$} (e);
        \draw (f) edge[below, sloped] node {\footnotesize $8,9$} (c);
        \draw (e) edge node {\footnotesize $7$} (f);
        \end{scope}
    \end{tikzpicture}
    \\

    $\mathcal{G}$ & $G_1$ & $\mathcal{G}_1$ & $\mathcal{H}$ \\
  \end{tabular}

  \caption{Dilation of the labels of $\mathcal{G}$. Here, a single snapshot, namely $G_1$ contains paths whose length is larger than $1$, so the dilation is only applied to $G_1$. The transformed snapshot $\G_1$ has $6$ labels (instead of $1$). It is then recomposed with the other snapshots, whose time labels are shifted accordingly (by $5$ time unit), resulting in graph $\mathcal{H}$.}
  \label{fig:shift-method}
\end{figure}
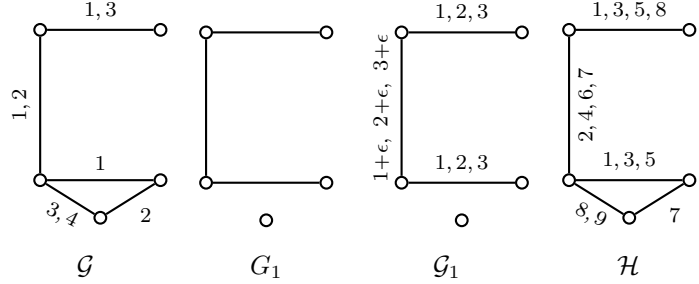

\begin{lemma}[Dilation technique]
  \label{lemma:time-split}
  Given a temporal graph $\G$, the dilation technique transforms $\mathcal{G}$ into a proper temporal graph $\mathcal{H}$ such that there is a non-strict journey in $\mathcal{G}$ if and only if there is a strict journey in $\mathcal{H}$ with same support.
\end{lemma}

\begin{proof}
  (Properness.) Every snapshot is dealt with independently and chronologically, the subsequent snapshots being shifted as needed to occur after the transformed version of the formers, so each snapshot becomes a temporal subgraph of $\mathcal{H}$ whose lifetime occupies a distinct subinterval of the lifetime of $\mathcal{H}$. Moreover, the tilting method based on a proper coloring of the edges guarantees that each of these temporal graphs is proper. Thus, $\mathcal{H}$ is proper. \\
  (Preservation of journeys.) Given a snapshot $G_t$ considered independently, the longest path in $G$ has length $k \le |V|-1$ and every edge has all the labels from $1$ to $k$, so for every path of length $\ell$ in $G$, there is a \emph{strict} journey in this graph, along the same sequence of edges, going over labels $1, 2, ..., \ell$ (up to the tilts, which are all less than $1$). Moreover, if a journey exists in $\G_t$, then its underlying path also exists in $G_t$ (since $G_t$ is the footprint of $\G_t$), thus, the dilation of a snapshot is support-preserving.
  Finally, as all the snapshots occupy a distinct subinterval of the lifetime of $\mathcal{H}$, and the order among snapshots is preserved, the composability of journeys over different snapshots is also unaffected.
\end{proof}

Let us clarify a few additional properties of the transformation. In particular,

\begin{lemma}
  The running time of the dilation technique is polynomial. 
\end{lemma}

\begin{proof}
  For any reasonable representation of $\G$ in memory, one can easily
  isolate a particular snapshot by filtering the contacts for the
  corresponding label (if the representation is itself snapshot-based,
  this step is even more direct). Then, each snapshot has at most
  $O(n^2)$ edges, so assigning the $n$ required labels to each of them
  takes at most $O(n^3)$ operations. Using Misra and Gries coloring
  algorithm~\cite{MG92}, a proper coloring of the edges can be
  obtained in $O(n^3)$ time per snapshot. 
  Applying the tilt operation takes
  essentially one operation per label in the transformed snapshot, so
  $O(n^3)$ again. Finally, these operations must be performed for
  all snapshots, which incurs a additional global factor of $\tau$,
  resulting in a total running time of essentially $O(n^3\tau)$.
  (The exact running time may depend on the actual data structure. It could also be characterized more finely by considering the number of contacts~$k$ (temporal edges) as a parameter instead of the rough approximation above, leading to a complexity of $O(kn)$ time steps, e.g. with adjacency lists to describe the snapshots.)
\end{proof}

Finally, observe that the dilation technique may incur a significant
blow up in the lifetime of the graph. More precisely,

\begin{lemma}
  Let $\Delta$ be the maximum degree of a vertex in any of the snapshot.
  Then, $\tau_{\mathcal{H}} \leq \tau_{\G}(\Delta+1)(n-1)$.
\end{lemma}

\begin{proof}
  There are $\tau_\G$ snapshots, each one can be turned into a
  temporal graph that uses up to $n-1$ nominal labels tilted in
  $\Delta+1$ different ways.
\end{proof}


\subsubsection{Saturation technique: ``non-strict'' $\to$ ``strict''}

As already explained, the dilation transformation described above makes it possible to transform any temporal graph in the non-strict setting into a proper graph, which is \emph{de facto} included in both the strict and non-strict setting. Thus, it can be seen as a support-preserving transformation from the non-strict setting to the strict setting, at the cost of a significant blow-up of the lifetime. In this section, we present a weaker transformation called the \emph{saturation} method. This transformation is only reachability-preserving, but it keeps the lifetime constant. An additional benefit is that it is pretty simple. A similar technique was used in~\cite{BCC+14} to test the temporal connectivity of a temporal graph with non-strict journeys.

\begin{theorem}
  Let $\mathcal{G}$ be a temporal graph with $n$ vertices, $m$ contacts and a lifetime of size~$\tau$, considered in the non-strict setting. There exists a temporal graph $\mathcal{H}$ with $n$ vertices, lifetime $\tau$ and at most $(n(n+1)\tau) / 2$ contacts in the strict setting that results in the same reachability graph.
\end{theorem}

\begin{proof}
  Let $\G$ be seen as a sequence of snapshots $G_1,..., G_\tau$. The transformation consists of transforming independently every snapshot $G_i$ of $\mathcal{G}$ by turning every path of $G_i$ into an edge. In other words, turning each snapshot $G_i$ into its own (path-based) transitive closure. The resulting graph $\mathcal{H}=H_1,..., H_\tau$ has the same lifetime as $\G$ and the same set of vertices.
  We will now prove that there is a (non-strict) journey from $u$ to $v$ in $\mathcal{G}$ if and only if there is a strict journey from $u$ to $v$ in $\mathcal{H}$.

    ($\to$) Let $j$ be a journey in $\mathcal{G}$. If any part of $j$ uses consecutive edges at the same time step $t$ (say, from $a$ to $b$), then there is a corresponding path in the snapshot $G_t$, implying an edge $ab$ in $H_t$, thus, this part of $j$ can be replaced by a contact $(\{a,b\},t)$ in $\mathcal{H}$. Repeating the argument implies a strict journey.

    ($\leftarrow$) Let $j'$ be a journey in $\mathcal{H}$. By construction of $\mathcal{H}$, for any contact $(\{a,b\},t)$ in $j'$, either the same contact already exists in $\G$, or there exists a path between $a$ and $b$ in $G_t$. If non-strict journeys are allowed, this path can replace the contact. Repeating the argument implies a non-strict journey.
\end{proof}

\subsubsection{Semaphore technique: ``strict'' $\to$ ``simple \& proper'' (happy)}

In this section, we describe a transformation called the \emph{semaphore technique}, which transforms any graph in the strict setting into a happy graph, while preserving the reachability between the vertices of the input graph, thus this transformation implies an induced-reachability equivalence.
  Our transformation is inspired by a reduction due to Bhadra and Ferreira~\cite{BF03} that reduces the \textsc{Clique} problem to the problem of finding maximum components in temporal graphs. However, their reduction takes as input a graph that is (morally) simple, and produces temporal graphs that are neither simple nor proper. Thus, our transformation differs significantly.

  \begin{theorem}
    \label{th:semaphore}
Let $\mathcal{G}$ be a temporal graph with $n$ vertices and $m$ contacts in the strict setting. There exists a happy graph $\mathcal{H}$ with $n+2m$ vertices and $4m$ edges and a mapping $\sigma: V_\G \to V_\mathcal{H}$ such that $(u,v) \in \closure(\G)$ if and only if $(\sigma(u), \sigma(v)) \in \closure(\mathcal{H})$.
\end{theorem}

\begin{proof}
  Intuitively, the transformation consists of turning every contact of $\G$, say $x = (\{u,v\},t_x)$, into a ``semaphore'' gadget in $\mathcal{H}$ that consists of a copy of $u$ and $v$, plus two auxiliary vertices $u_x$ and $v_x$ linked by~$4$ edges $\{u,u_x\}$, $\{u_x,v\}$, $\{u,v_x\}$, $\{v_x,v\}$, whose labels create a journey from $u$ to $v$ through $u_x$, and from $v$ to $u$ through $v_x$. The labels are chosen in such a way that these journeys can replace $x$ for the composition of journeys in $\mathcal{H}$. For simplicity, our construction uses fractional label values, which can subsequently be renormalized into integers. A basic example is shown in Figure~\ref{fig:diamond-method}.

\begin{figure}[ht]
  \centering
  \begin{tabular}{cc}
    \begin{tikzpicture}[every loop/.style={}, thick, node distance=1.2cm,auto]
        \tikzset{vertex/.style={draw, circle, inner sep=0.55mm}}

        \node[vertex] (u) [label = $u$] at (0,0) {};
        \node[vertex] (v) [label = $v$] at (1,0) {};
        \node[vertex] (w) [label = $w$] at (2,0) {};

        \begin{scope}
        \draw (u) edge node {1,2} (v);
        \draw (v) edge node {1} (w);
        \end{scope}
    \end{tikzpicture}
    &
    \begin{tikzpicture}[every loop/.style={}, thick, node distance=1.2cm,auto]
        \tikzset{vertex/.style={draw, circle, inner sep=0.55mm}}

        \node[vertex] (u) [label = left:$u$] at (0,0.5) {};
        \node[vertex] (u1) at (1.75,-1) {};
        \node[vertex] (u1') at (1.75,-0.4) {};
        \node[vertex] (u2) at (1.75,1.4) {};
        \node[vertex] (u2') at (1.75,2) {};
        \node[vertex] (v) [label = $v$] at (3.5,0.5) {};
        \node[vertex] (v1) at (4.75,0) {};
        \node[vertex] (v1') at (4.75,1) {};
        \node[vertex] (w) [label = right:$w$] at (6,0.5) {};

        \begin{scope}
          \draw (u) edge node[below, sloped] {\footnotesize $1-\epsilon$} (u1);
          \draw (u1) edge node[below, sloped] {\footnotesize  $1+\epsilon$} (v);
          \draw (u) edge node[above, sloped] {\footnotesize $1+\epsilon$} (u1');
          \draw (u1') edge node[above, sloped] {\footnotesize $1-\epsilon$} (v);

          \draw (u) edge node[below, sloped] {\footnotesize $2-\epsilon$} (u2);
          \draw (u2) edge node[below, sloped] {\footnotesize $2+\epsilon$} (v);
          \draw (u) edge node[above, sloped] {\footnotesize $2+\epsilon$} (u2');
          \draw (u2') edge node[above, sloped] {\footnotesize $2-\epsilon$} (v);

          \draw (v) edge[below, sloped] node {\footnotesize $1-2\epsilon$} (v1);
          \draw (v1) edge[below, sloped] node {\footnotesize $1+2\epsilon$} (w);
          \draw (v) edge node[above, sloped] {\footnotesize $1+2\epsilon$} (v1');
          \draw (v1') edge node[above, sloped] {\footnotesize $1-2\epsilon$} (w);
        \end{scope}
    \end{tikzpicture}
    \\
    $\mathcal{G}$ & $\mathcal{H}$ \\
  \end{tabular}
  \caption{The semaphore technique, turning a non-proper graph $\mathcal{G}$ (in the strict setting), into a happy graph $\mathcal{H}$ whose reachability preserves the relation among original vertices.}
  \label{fig:diamond-method}
\end{figure}
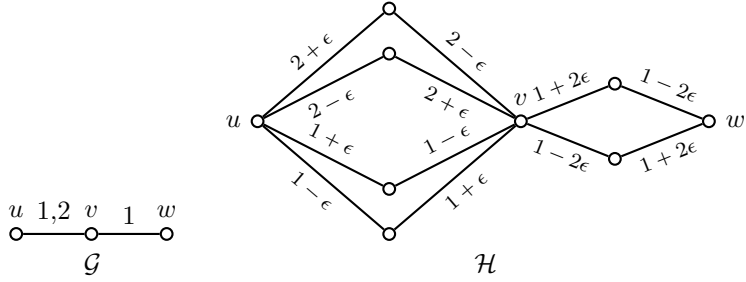

Since we consider strict journeys in $\mathcal{G}$, if two adjacent time edges share the same time label, it should be forbidden to take the two of them consecutively. To ensure this, the time labels of $\{u,u_x\},\{u_x,v\},\{u,v_x\},\{v_x,v\}$ are respectively $t_x - \epsilon$, $t_x + \epsilon, t_x + \epsilon$, $t_x - \epsilon$ (with $0 < \epsilon < 1/2$), enabling the journeys from $u$ to $v$ and from $v$ to $u$ without making two such journeys composable if the two original labels are the same. The created edges have a single label, but the graph might not yet be proper. To ensure properness, we tilt slightly the time labels by multiples of $\epsilon$, in a similar spirit as in the dilation technique presented above. More precisely,
consider a proper edge-coloring of the footprint of $\mathcal{G}$ using $\Delta+1$ colors in $\{1,...,\Delta+1\}$ (where $\Delta$ is the maximum degree in the footprint), such a coloring being guaranteed by Vizing's theorem. For each edge $e$ of the footprint, note $c_e$ its color. Now the time labels in $\mathcal{H}$ associated to $x = (e,t_x)$ are $t_x - c_e \epsilon$, $t_e + c_e \epsilon$ with $0 < \epsilon < \frac{1}{2(\Delta+1)}$.

The semaphore gadget is applied for each contact of $\G$ with the corresponding color, so if there was $n$ vertices and $m$ time edges in $\mathcal{G}$, then $\mathcal{H}$ will have $n+2m$ vertices and~$4m$ time edges. It is easy to see that $\mathcal{H}$ is now simple and proper.
We will now prove that, for any $u$ and $v$,  $u$ can reach $v$ in $\G$ if and only if $u$ can reach $v$ in $\mathcal{H}$.

($\to$) First note that for any pair of adjacent vertices $u,v\in V_{\G}$, we can go from $u$ to $v$ in $\mathcal{H}$ by following one side of the "semaphore" and from $v$ to $u$ by following the other side. Furthermore, for any two $u,v\in V_\G$ such that there is a strict journey from $u$ to $v$ in $\mathcal{G}$, there is a sequence of edges $e_1,e_2,\dots, e_k$ with labels $t_1,t_2,\dots, t_k$ in $\mathcal{G}$ such that $e_1$ is incident to $u$, $e_k$ is incident to $v$, $e_i$ is adjacent to $e_{i+1}$ and $t_i<t_{i+1}$. Then $u$ can reach $v$ in $\mathcal{H}$ by a sequence of edges with increasing labels $t_1-c_{e_1}\ep, t_1+c_{e_1}\ep, t_2-c_{e_2}\ep, t_2+c_{e_2}\ep,\dots, t_k-c_{e_k}\ep, t_k+c_{e_k}\ep$.

($\leftarrow$) First note that from any original vertex $a$, seen in $\mathcal{H}$, we must first move to an auxiliary vertex $b$ by the edge $e_1=\{a,b\}$ and then to another original vertex $c$ by edge $e_2=\{b,c\}$. The labels on the edges must be $t-c\ep$ and $t+c\ep$ for some $t$, where $t$ is the original time of the contact of $\mathcal{G}$ and $c$ the color of the underlying edge. Then, we again reach an auxiliary vertex $d$ from $c$ at time $t'-c'\ep$, where $t+c\ep<t'-c'\ep$, and then another original vertex $e$ at time $t'+c'\ep$. Since $t<t'$, we can go from $a$ to $e$ through $c$ in $\mathcal{G}$ using contacts $(\{a,c\}, t)$ and $(\{c,e\},t')$, respectively. Hence for all original $u$ and $v$, if we have a temporal path from $u$ to $v$ in $\mathcal{H}$, then we can follow the above process multiple times to get a temporal path from $u$ to $v$ in $\mathcal{G}$.
\end{proof}



Observe that the semaphore technique is, in a quite relaxed way, also support-preserving, in the sense that a journey in $\G$ can be mapped into a journey in $\mathcal{H}$ that traverses the \emph{original} vertices in the same order (and vice versa), albeit with auxiliary vertices in between.


\subsection{Summary and discussions}
\label{sec:summary}

Let $S_1$ and $S_2$ be two different settings, we define an order
relation $\preceq$ so that $S_1 \preceq S_2$ means that for any graph
$\G_1$ in $S_1$, one can find a graph $\G_2$ in $S_2$ such that
$\closure(\G_1) \simeq \closure(\G_2)$. We write $S_1 \precneq S_2$ if
the containment is proper (i.e., there is a graph in $S_2$ whose
reachability graph cannot be obtained from a graph in $S_1$). Finally,
we write $S_1 \approx S_2$ if both sets of reachability graphs
coincide. Several relations follow directly from containment among
graph classes, e.g. the fact that simple graphs are a particular case
of non-simple graphs. The above separations and transformations also
imply a number of relations, and their combination as well. For
example, proper graphs are contained both in the strict and non-strict
settings, and since there
is a transformation from non-strict graphs (in general) to proper
graphs, we have the following striking relation:

\begin{corollary}
  \label{corollary:ft-proper-non-strict}
  ``Proper'' $\approx$ ``non-strict''.
\end{corollary}



Similarly, combining the fact that ``simple \& non-strict'' is strictly contained in ``non-strict'' (by Corollary~\ref{cor:simple<non-strict}), and there exists a reachability-preserving (in fact, support-preserving) transformation from ``non-strict'' to ``proper'', we also have that

\begin{corollary}
  \label{cor:simple-nstrict<proper}
  ``Simple \& non-strict'' $\precneq$ ``proper''.
\end{corollary}


Finally, the fact that there is a reachability-preserving transformation from ``non-strict'' to ``strict'' (the saturation technique), and some reachability graphs from ``simple \& strict'' are unrealizable in ``non-strict'' (by Lemma~\ref{lemma:ft-strict-non-strict}), we also have

\begin{corollary}
  \label{cor:nstrict<strict}
  ``Non-strict'' $\precneq$ ``strict''.
\end{corollary}


A summary of the relations is shown in Figure~\ref{fig:summary}, where
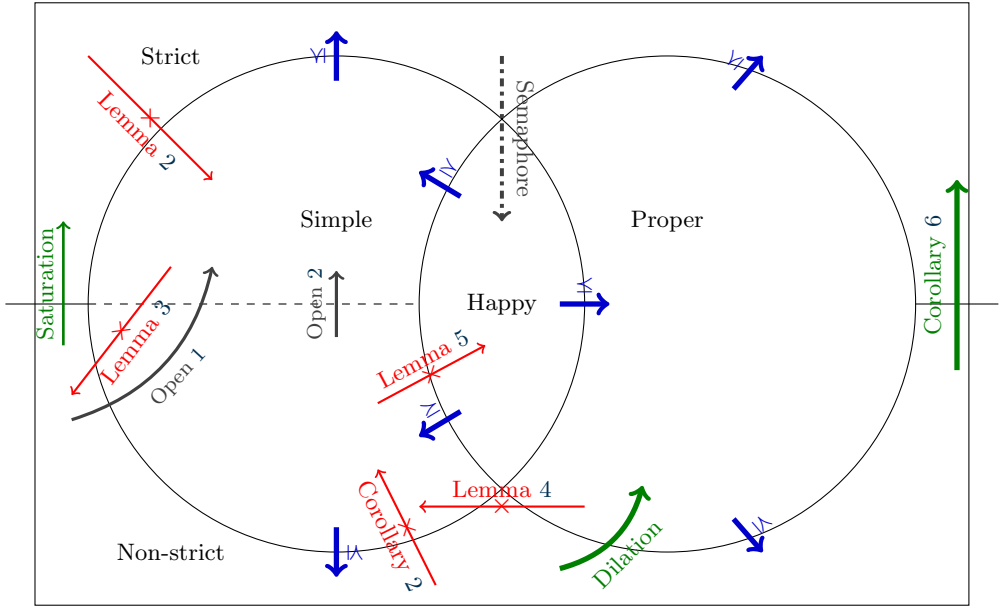
\begin{figure}[h]
\centering
\colorlet{circle edge}{blue!50}
\colorlet{circle area}{blue!20}

\tikzset{
  filled/.style={fill = circle area, draw=circle edge, thick},
  outline/.style={draw=circle edge, thick}
}

\begin{tikzpicture}[scale=1.1]

  \draw (0,0) circle[radius=3cm]
  (0:4cm) circle[radius=3cm];
  \tikzstyle{every node}=[font=\small]
  \tikzstyle{separation} = [red,->]
  \tikzstyle{saturation} = [black!50!green,line width=1pt,->]
  \tikzstyle{inclusion} = [black!20!blue,line width=2pt,->]
  \tikzstyle{semaphore} = [darkgray, line width=1.5pt, dashdotted, ->]
  \tikzstyle{splitting} = [black!50!green, line width=2pt,->]
  \tikzstyle{question} = [darkgray, line width=1.2pt,->]

  \node [draw,fit=(current bounding box),inner sep=7mm] (frame){};

  \node at (0:2cm) {Happy};
  \node at (0, 1) {Simple};
  \node at (4, 1) {Proper};
  \node at (-2, 3) {Strict};
  \node at (-2,-3) {Non-strict};

  \draw (-4,0) -- (-3,0);
  \draw[dashed] (-3,0) -- (1,0);
  \draw (7,0) -- (8,0);

  \tikzstyle{every path}=[thick]
  \draw (-3,3) edge[separation] node[sloped]{\large $\times$} node[below, sloped]{Lemma~\ref{lemma:ft-strict-simple}} (-1.5,1.5);
  \draw (3,-2.45) edge[separation] node[sloped]{\large $\times$} node[above]{Lemma~\ref{lemma:ft-non-strict-simple}} (1,-2.45);
  \draw (-2,0.45) edge[separation] node[sloped]{\large $\times$} node[sloped, below]{Lemma~\ref{lemma:ft-strict-non-strict}} (-3.2,-1.1);
  \draw (0.5,-1.2) edge[separation] node[sloped]{\large $\times$} node[sloped,above]{Lemma~\ref{lemma:ft-non-strict-simple-happy}} (1.8,-0.5);

  \draw (1.2,-3.4) edge[separation] node[sloped]{\large $\times$} node[sloped,below]{\small Corollary~\ref{cor:simple<non-strict}} (0.5,-2);

  \draw (2,3) edge[semaphore] node[above, sloped] {Semaphore} (2,1);
  \draw (2.7,-3.2) edge[splitting, bend right] node[below, sloped] {Dilation} (3.7,-2.2);
  \draw (-3.3,-0.5) edge[saturation] node[sloped, above]{Saturation} (-3.3,1);

  \draw (7.5,-0.8) edge[splitting] node[sloped,above]{\small Corollary~\ref{cor:nstrict<strict}} (7.5,1.5);

  \draw (0,2.7) edge[inclusion] node[above, sloped]{$\preceq$} (0,3.3);
  \draw (0,-2.7) edge[inclusion] node[above, sloped]{$\preceq$} (0,-3.3);
  \draw (4.8,2.6) edge[inclusion] node[above, sloped]{$\preceq$} (5.15,3);
  \draw (4.8,-2.6) edge[inclusion] node[above, sloped]{$\preceq$} (5.15,-3);
  \draw (2.7,0) edge[inclusion] node[above, sloped]{$\preceq$} (3.3,0);
  \draw (1.5,1.3) edge[inclusion] node[above, sloped]{$\succeq$} (1,1.6);
  \draw (1.5,-1.3) edge[inclusion] node[above, sloped]{$\succeq$} (1,-1.6);

  \draw (0,-0.4) edge[question] node[above,sloped]{\footnotesize Open~\ref{open:simple-nstrict<simple-strict}} (0,0.4);
  \draw (-3.2,-1.4) edge[bend right, question] node[below,sloped]{\footnotesize Open~\ref{open:nstrict<simple-strict}} (-1.5,0.45);

\end{tikzpicture}


\vspace{-5pt}
\caption{\label{fig:summary}Separations, transformations, and
  inclusions among settings.}
\end{figure}
green thick edges represent the transformations that are
support-preserving, green thin edges represent transformations that
are reachability-preserving, red edges with a cross represent separations
(i.e. the impossibility of such a transformation), black dashed edges
represent the induced-reachability-preserving transformation to happy
graphs. Finally, inclusions of settings resulting from containment of
graph classes are depicted by short blue edges.
Some questions remain open. In particular,

\begin{open}
  \label{open:nstrict<simple-strict}
  Does ``non-strict'' $\preceq$ ``simple \& strict''? In other words, is there a reachability-preserving transformation from the former to the latter?
\end{open}

By Lemma~\ref{lemma:ft-strict-non-strict}, we know that both settings are not equivalent, but are they comparable? If not, a similar question holds for ``simple \& non-strict'':

\begin{open}
  \label{open:simple-nstrict<simple-strict}
  Does ``simple \& non-strict'' $\preceq$ ``simple \& strict''? In other words, is there a reachability-preserving transformation from the former to the latter?
\end{open}

To conclude this section, Figure~\ref{fig:order-temporal-graphs} depicts a hierarchy of the settings ordered by the above relation $\preceq$; i.e. by the sets of reachability graph they can achieve.
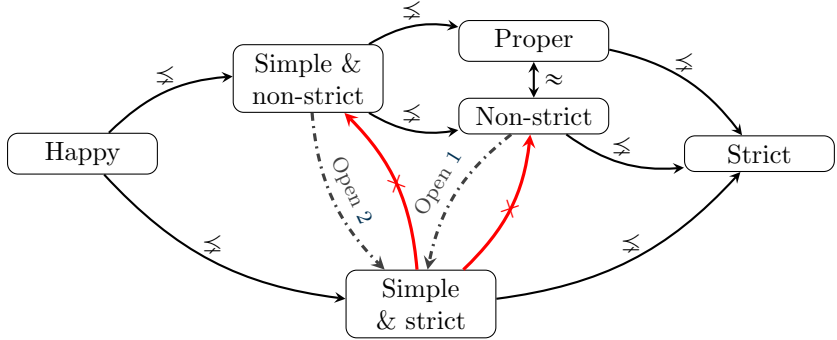
\begin{figure}[h]
  \centering
  \begin{tikzpicture}[scale=1]
  \tikzstyle{separation} = [red,line width=1.2pt,->]
  \tikzstyle{question} = [darkgray, dashdotted, line width=1.2pt,->]

  \tikzset{vertex/.style={draw, rectangle,
           rounded corners, draw=black, text centered, text width=5em}}

  \node[vertex] (happy) at (-0.5,0) {Happy};
  \node[vertex] (simple_nstrict) at (2.5, 1) {Simple \& non-strict};
  \node[vertex] (nstrict) at (5.5, 0.5) {Non-strict};
  \node[vertex] (proper) at (5.5, 1.5) {Proper};
  \node[vertex] (strict) at (8.5,0) {Strict};
  \node[vertex] (simple_strict) at (4,-2) {Simple \& strict};

  \tikzstyle{every path}=[thick]
  \begin{scope}[->,>=stealth,bend angle=20]
    \draw (happy) edge[bend left] node[above] {$\precneq$} (simple_nstrict);
    \draw (happy) edge[bend right] node[above] {$\precneq$} (simple_strict);
    \draw (simple_nstrict) edge[bend left, bend angle=5] node[above] {$\precneq$} (proper);
    \draw (simple_nstrict) edge[bend right, bend angle=5] node[above] {$\precneq$} (nstrict);
    \draw (proper) edge[bend left, bend angle=5] node[above] {$\precneq$} (strict);
    \draw (nstrict) edge[bend right, bend angle=5] node[above] {$\precneq$} (strict);

    \draw (nstrict) edge[<->] node[right] {$\approx$} (proper);
    \draw (simple_strict) edge[bend right] node[above] {$\precneq$} (strict);

    \draw (simple_strict) edge[bend right,separation] node[sloped]{\large $\times$} (simple_nstrict);
    \draw (simple_nstrict) edge[bend right,question] node[above,sloped] {\footnotesize Open~\ref{open:simple-nstrict<simple-strict}} (simple_strict);
    \draw (simple_strict) edge[bend right,separation] node[sloped]{\large $\times$} (nstrict);
    \draw (nstrict) edge[bend right,question] node[above,sloped]{\footnotesize Open~\ref{open:nstrict<simple-strict}} (simple_strict);
  \end{scope}

\end{tikzpicture}


  \caption{Ordering of temporal graph settings in terms of reachability graphs.}
  \label{fig:order-temporal-graphs}
\end{figure}

\section{Strengthening existing results to happy graphs}
\label{sec:happy}

From the previous section, happy graphs
are the least expressive setting. In this section, however, we argue that
they remain expressive enough to strengthen existing negative results
for at least two well-studied problems. First, we show that the construction
from~\cite{AF16} can be made happy, which implies that $o(n^2)$-sparse
spanners do not always exist in happy graphs. We also show that the
reduction from clique to temporal component from~\cite{BF03} can be made
happy, which implies that temporal component is NP-complete
even in happy graphs (for both open and closed components).
Finally, to further motivate studies on happy graphs, we list a few open questions.

\subsection{\textsc{Temporal Component} remains hard}
\label{sec:maxcomponent}

The temporal component problem can be defined as follows in temporal graphs.

\begin{problem}{\textsc{Temporal Component}}
  Input: A temporal graph $\G=(V,E,\lambda)$, an integer $k$\\
  Output: Is there a set $V'\subseteq V$ of size $k$ in $\G$ such that all vertices of $V'$ can reach each other by a temporal path?
\end{problem}

Due to the non-transitive nature of reachability, two versions
are typically considered, depending on whether
the vertices of $V'$ can rely (open version) or not (closed version)
on vertices outside of $V'$ for reaching each other.

In~\cite{BF03}, Bhadra and Ferreira show that both versions of the
problem are NP-complete. Interestingly, although the authors consider
a strict setting in the paper, their construction works indistinctly
for both the strict or the non-strict setting. However, it is neither
proper nor simple. In this section, we explain how to adapt their
construction to happy graphs. Our reduction works for both the open
and closed version, for the simple fact that it produces an instance
where the maximum open and closed components are the same.

Let $(G,k)$ be an instance of the clique problem, where $G$ is a
static undirected graph and $k$ some integer, the reduction
in~\cite{BF03} transforms $G$ into a temporal graph $\G$ as follows.
The first step of the transformation corresponds to a simplified
version of the semaphore technique, where two auxiliary vertices are
created for each pair of neighbors $u$ and $v$ in $G$, such that $u$
can reach $v$ in $\G$ through one of these vertices and $v$ can reach
$u$ through the other, using labels $2$ on the first edge and $3$ on
the second (on both sides). Observe that two adjacent semaphores
have labels which are not proper. The second step is to connect all
pairs of auxiliary vertices $x$ and $y$ using an edge $xy$ with two
labels~$1$ and~$4$ for each pair (thus $\G$ is not simple). The
purpose of these contacts is to make all auxiliary vertices reachable
from each other, and to create journeys between each auxiliary vertex
and each original vertex (both ways). As a result, an SCC of size
$2m + k$ exists in $\G$ (where $m$ is the number of edges of~$G$) if
and only if a clique of size $k$ exists in $G$.

\begin{theorem}
  \textsc{Temporal Component} is NP-complete in happy graphs.
\end{theorem}

\begin{proof}
  Our proof consists of making the transformation from~\cite{BF03}
  both simple and proper, while preserving the size of the maximum
  component within (which is the same for the open and closed versions).
  First, observe that the non-proper labels of the
  semaphores in $\G$ can easily be turned into proper labels by
  \emph{tilting} the labels in the same way as explained in the
  \emph{semaphore technique} (Theorem~\vref{th:semaphore}), namely, by
  coloring properly the edges of the footprint of $\G$, and adding the
  corresponding multiple of $\epsilon$ to every label. Since the
  quantity added to each label is less than $1$, the reachability
  among original vertices, and between original and auxiliary
  vertices, is unaffected. The conversion of the labels $1$ and $4$
  between auxiliary vertices is slightly more complicated. These
  vertices in $\G$ form a clique, each edge of which have labels $1$
  and $4$. First, using Lemma B.1 in~\cite{AF16}, any clique of $2m$
  vertices may be decomposed into $m$ hamiltonian paths. We need only
  $4$ such paths for the construction, which is guaranted as soon as
  $m \ge 4$ (our adaptation does not need to hold for smaller graphs
  in order to conclude that the problem is NP-complete). Thus, take
  four edge-disjoint hamiltonian paths among auxiliary vertices. We
  will use two of them, say $p_1$ and $p_2$ to replace the contacts
  having label $1$ (the same technique applies for label $4$). Pick a
  vertex $u$ in $p_1$ and assign time labels to $p_1$ so that all
  vertices in $p_1$ can reach $u$ through ascending labels (towards
  $u$). Then proceed similarly in $p_2$ with ascending labels
  \emph{from} $u$ towards all the vertices of $p_2$. Choose these
  labels so that they remain sufficiently close to $1$ and do not
  interfere with the rest of the construction. Proceed similarly for
  the two other hamiltonian paths with respect to label $4$. The
  resulting construction is happy and preserves the journeys between
  auxiliary vertices, while preserving the composability of journeys
  with the rest of the construction.
\end{proof}

\subsection{Happy graphs do not always admit $o(n^2)$-sparse spanners}
\label{sec:counterexample}

In~\cite{AF16} (Theorem~3.1), Axiotis and Fotakis construct an infinite family of temporal graphs in the ``simple \& non-strict'' setting that does not admit a $o(n^2)$-sparse spanner. The goal of this section is to show that their construction can be strengthened to happy graphs.

The construction $\mathcal{G}$ in~\cite{AF16} consists of three parts of $n$ vertices each (thus $N=3n$ vertices in total), namely a clique of vertices $A=a_1,...,a_n$; an independent set $H=h_1,...,h_n$; and a set $M=m_1,...,m_n$ of additional vertices. The idea is to make every edge of $A$ critical to provide connectivity among some vertices of $H$, so that removing any of these edge breaks temporal connectivity and every spanner thus contains $\Theta(n^2)$ edges. The purpose of the vertices in $M$ is only to make the rest of $\mathcal{G}$ temporally connected without affecting these relations between $A$ and $H$. In this construction, every edge receives a single label, so $\G$ is already simple. However, the inner labeling of the clique $A$ is not proper. We claim that this labeling can be made proper without affecting the main properties of the construction. The following statement is identical to Theorem~3.1 in~\cite{AF16}, except that the adjective happy is inserted.

\begin{theorem}
  \label{th:spanners}
For any even $n\geq 2$, there is a happy connected temporal graph with $N=3n$ vertices, $\frac{n(n+9)}{2}-3$ edges and lifetime at most $\frac{n(n+5)}{2}-1$, so that the removal of any subset of $5n$ edges results in a disconnected temporal graph.
\end{theorem}

\begin{proof}[Proof (sketch)]
  The complete construction from~\cite{AF16} is not presented here in detail. However, the fact that it can be made proper relies on a simple observation. In~\cite{AF16}, the clique $A$ is decomposed into $\frac{n}{2}$ hamiltonian paths $p_1, p_2, ..., p_{n/2}$, each of which is assigned label $i$.
For every path $p_i$, vertices $h_{2i-1}$ and $h_{2i}$ are connected to the endpoints of $p_i$ (one on each side), and the main requirement is that this path is the \emph{only} way for $h_{2i-1}$ to reach $h_{2i}$. Interestingly, although every path $p_i$ is non-strict in~\cite{AF16} and thus could be travelled in both directions, it turns out that only one direction is needed, because $h_{2i}$ can reach $h_{2i-1}$ (for every $i$) using temporal paths outside of the clique. Therefore, our adaptation consists of assigning to every path $i$ a strictly increasing sequence of labels, while shifting all the larger labels of the graph appropriately, so that all other temporal paths are unaffected and $\G$ becomes proper. The full proof would require a complete description of the construction in~\cite{AF16}, which would be identical up to the above change.
\end{proof}

\subsection{Further questions}

To conclude this section, we state a few open questions related to
spanners in happy graphs. The first question is structural,
namely,

\begin{open}Do happy cliques always admit $O(n)$-sparse spanners? If so, do they admit spanners of size $2n-3$?
\end{open}

It was shown in~\cite{CPS21} that happy cliques (or alternatively, all temporal cliques in the non-strict setting) always admit spanners of size $O(n \log n)$, but no counterexamples were found so far that rule out size $O(n)$.

On the algorithmic side, two independent results establish that \textsc{Min-Label Spanner} is hard in temporal graphs~\cite{AF16,AGMS17}. However, the proofs in these papers rely on constructions which are not happy, and it is not clear that these constructions can be strengthened to happy graphs in a similar way as the above problems. Thus,

  \begin{open} Is \textsc{Min Spanner} tractable in happy graphs?
  \end{open}

As already explained, both the \minedge\ and \minlabel\ versions of the problem coincide in simple graphs (and thus in happy graphs), which makes them a single problem.

\section{Concluding remarks}
\label{sec:conclusion}

In this paper, we explored the impact of three particular aspects of
temporal graphs: \emph{strictness}, \emph{properness}, and
\emph{simpleness}. Comparing their expressivity in terms of reachability
graphs, we showed that these aspects really matter and that
separations exist between the expressivity of some
settings, while others can be shown equivalent through transformations.
Then, we focused on the simplest model (happy graphs), where
all these distinctions vanish, and showed that this model still
captures interesting features of general temporal graphs. Our results
imply a few striking facts, such as the fact that the ``proper'' setting is as
expressive as the ``non-strict'' setting. Some relations remain
unknown, in particular, it is open whether the ``non-strict'' setting
is comparable to ``simple \& strict'' setting. Finally, despite
their extreme simplicity, several basic questions remain open on happy
graphs, which we think makes them a natural target for further
studies. We conclude by stating a few questions of more general scope,
related to the present paper.

\begin{open}[Realizability of a reachability graph]
  \label{open:realizability}
  Given a static digraph, how hard is it to decide whether it can be realized as the reachability graph of a temporal graph?
\end{open}

The structural analog of this question could be formulated as follows

\begin{open}[Characterization of the reachability graphs]
  \label{open:characterization}
  Characterize the set of static directed graphs that are the reachability graphs of some temporal graph.
\end{open}

Questions~\ref{open:realizability}
and~\ref{open:characterization} can be declined into several versions, one for each setting.
Finally, the work in this paper focused on \emph{undirected} temporal graphs.
It would be interesting to see if the expressivity of \emph{directed} temporal graphs shows similar
separations and transformations.

\begin{open}[Directed temporal graphs]
  \label{open:directed-separation}
  Does the expressivity of directed temporal graphs admit similar
  separations and transformations as in the undirected case?
\end{open}

{\footnotesize
\bibliographystyle{plain}
\bibliography{happy.bib}
}

\newpage
\appendix

\end{document}